\documentclass[12pt]{article}

\usepackage{graphicx}
\usepackage{amssymb,amsmath,amsthm}
\usepackage{xcolor}

\makeatletter
\newcommand{\globallabel}[1]{%
  \protected@edef\@currentlabel{\theparentequation a--\alph{equation}}\label{#1}%
}
\makeatother

\newcommand{\e}{\mathrm{e}}

\newcommand{\N}{\mathbb{N}}
\newcommand{\R}{\mathbb{R}}

\newcommand{\Oo}{\mathcal{O}}

\newcommand{\Hm}[1]{\leavevmode{\marginpar{\tiny%
$\hbox to 0mm{\hspace*{-0.5mm}$\leftarrow$\hss}%
\vcenter{\vrule depth 0.1mm height 0.1mm width \the\marginparwidth}%
\hbox to
0mm{\hss$\rightarrow$\hspace*{-0.5mm}}$\\\relax\raggedright #1}}}

\newtheorem{claim}{Claim}[section]
\newtheorem{theorem}[claim]{Theorem}
\newtheorem{lemma}[claim]{Lemma}

\newtheorem{remark}[claim]{Remark}
\newtheorem{remarks}[claim]{Remarks}

\newtheorem{corollary}[claim]{Corollary}

\begin{document}

\title{Spectral optimization for strongly singular Schr\"odinger operators with a star-shaped interaction }

\author{P. EXNER$^{1,2}$ and S. KONDEJ$^{3}$}
\date{\small $^1$Doppler Institute for Mathematical Physics and Applied Mathematics, \\ Czech Technical University in Prague, B\v{r}ehov\'{a} 7, 11519 Prague, Czechia \\ $^2$Nuclear Physics Institute CAS, 25068 \v{R}e\v{z} near Prague, Czechia \\
$^3$Institute of Physics, University of Zielona G\'ora, ul.\ Szafrana 4a, \\ 65246 Zielona G\'ora, Poland \\ \emph{e-mail: exner@ujf.cas.cz, s.kondej@if.uz.zgora.pl}}

\maketitle

\begin{abstract}
\noindent We discuss the spectral properties of singular Schr\"odinger operators in three dimensions with the interaction supported by an equilateral star, finite or infinite. In the finite case the discrete spectrum is nonempty if the star arms are long enough. Our main result concerns spectral optimization: we show that the principal eigenvalue is uniquely maxi\-mized when the arms are arranged in one of the known five sharp configurations known as \color{black}{solutions of the closely related} Thomson problem.
\end{abstract}

\medskip


\medskip


\section{Introduction}\label{s:intro}
\setcounter{equation}{0}

Isoperimetric inequalities represent a traditional problem in mathematical physics with the first fundamental results almost a century old \cite{Fa23, Kr25}. Recent years witnessed a new wave of interest to them, for instance, in the context of Robin Laplacians, cf. \cite{FK15, KL17} and references therein. Another context in which such questions arise concerns singular Schr\"odinger operators which could be formally written as
 \begin{equation} \label{formal}
 H_{\alpha,\gamma} = -\Delta+\tilde\alpha\delta(x-\gamma)\,.
 \end{equation}
If $\gamma$ is a loop of a fixed length in the plane, e.g., it is known that the principal eigenvalue is maximized by a circle \cite{EHL06}. It three dimensions the problem is more complicated and decisive quantity is the capacity of $\gamma\,$ \cite{EF09}, note that similar result can be obtained for Dirac operators with a shell interaction \cite{AMV16}.

One can consider also other shapes of the interaction support. In \cite{EL17}, for instance, the support $\gamma$ in the shape of an equilateral planar star is discussed and it is proved that the principal eigenvalue is then maximized by the configuration of the maximum symmetry when all the angles between the neighboring star arms coincide. In the present letter we address the analogous question for three-dimensional Schr\"odinger operators which is considerably more complicated. One reason is the character of the singular interaction which is more singular if its support is of codimension two \cite{AGHH}. What is more important, however, is that the geometry of the star characterized by the distribution of its projection to the unit sphere is much richer, and consequently, the answer depends strongly on the number $N$ of the star arms; one may recall in the connection Thomson's problem \cite{Th1904} still not fully solved more than a century after it was formulated. We manage to show that the principal eigenvalue of the corresponding singular Schr\"odinger operator is uniquely maximized by the known sharp configurations \cite{CK07} for $N=2,3,4,6,$ and $12$, leaving a lot of room for investigation of stars with other values of~$N$.

The contents and the main results of the paper can be summarized  as follows. For a finite star we prove in~Sec.~\ref{s:existence} the existence of the discrete spectrum provided the $\delta$ interaction is sufficiently strong, and for an infinite star, if the support of interaction does not coincide with a straight line. On the other hand, if the interaction in the finite case is weak enough the discrete spectrum {\color{black}is void as will be proved in} Sec.~\ref{s:nonexistence}. Furthermore, in~Sec.~\ref{s:small} we show that there is no minimum since the threshold  of the spectrum can be arbitrarily low for small enough angle between a pair of arms. Finally, in Sec.~\ref{s:optim} we turn to the main topic and demonstrate the above mentioned configurations optimizing the principal eigenvalue.

\section{Preliminaries}\label{s:prelim}
\setcounter{equation}{0}

Our first task is to give a proper meaning to the formal operator \eqref{formal}. In general, the way how to do that is known -- cf.~\cite{EK02, EK08, EK16} and references there -- so we can focus on properties associated with the particular shape of the interaction support.

First we have to introduce some notation. Given $L\in(0,\infty]$, finite or infinite, we consider a family of $N$ line  segments, being the graphs of linear functions $\gamma_i:\: [0,L]\to \R^3$, emanating from the same point $\gamma_i(0)$ which can be without loss of generality set as the coordinate origin. With an abuse of notation we identify the edges with the functions $\gamma_i:\: [0,L]\to\R^3$ that parametrize them. It is clear that up to Euclidean transformations each such star is uniquely determined by the intersections $\bar\gamma_i$ of $\gamma_i$ (or their line extensions) with the unit sphere $S^2$ centered at the origin. The geometric quantity which will be important in the following is the distance between a pair of points of $\gamma$ which is expressed in terms of the used parameters as
 \begin{equation} \label{distance}
 |\gamma_i(s) - \gamma_j(t)|^2= s^2+t^2 - st(2 -|\bar\gamma_i - \bar\gamma_j |^2)\,.
 \end{equation}

The most direct way to define the operator of our interest is to impose suitable boundary conditions in cross planes to the arms $\gamma_i$, namely those that determine the two-dimensional point interaction in the plane with a parameter $\alpha\in\R\,$ \cite[Chap.~I.5]{AGHH}. Recall that the corresponding Hamiltonian has a single negative eigenvalue
 \begin{equation} \label{2D ev}
\epsilon^\alpha = -4\,\mathrm{e}^{2(-2\pi\alpha +\psi(1))}\,,
 \end{equation}
where $\psi$ is the digamma function and $-\psi(1)\approx 0.577$ is the Euler-Mascheroni constant. Given $f\in W^{2,2}_{\mathrm{loc}} (\R^3\setminus\gamma)$ we pick a point $s\in\gamma_i$ and its circular flat neighborhood $U_i$ in the plane perpendicular to $\gamma_i$ which is additionally assumed to be disjoint with $\gamma\setminus\gamma_i$; with the exception of the star vertex this can be always achieved provided $\rho = \rho(s)$, the radius of $U_i$, is small enough. Furthermore, let us consider the restriction $f\!\upharpoonright_{U_i}$ which is locally, that is in $U_i$, a distribution. We assume that the limits
\begin{subequations}
\begin{align}
\label{eq-bc1} &\Xi(f)(s) :=- \lim_{\rho \to 0 }\, \frac{1}{\ln \rho }
   f\!\upharpoonright_{U_i}\!\!(s)\,,
\\ \label{eq-bc2} & \Omega (f) (s):= \lim_{\rho \to 0 } \big[
f\!\upharpoonright_{U_i}\!\!(s) + \Xi(f)(s)\ln  \rho \big]\,
\end{align}
\globallabel{eq-bc}
\end{subequations}

\noindent exist almost everywhere in $(0,L)$ for any $i=1,\dots,N$. Imposing then the boundary con\-di\-ti\-ons coupling these generalized boundary values,
\begin{equation}\label{eq-bc3}
2\pi \alpha \Xi (f)=\Omega (f)\,,
\end{equation}
we get a self-adjoint operator $H_{\alpha,\gamma}$ with the domain
$$
D(H_{\alpha,\gamma}):= \{f\in W^{2,2}_{\mathrm{loc}} (\R^3\setminus \gamma )\cap L^2 \,:\, f\;\text{satisfies}\; \eqref{eq-bc}\}
$$
which acts as
$$
H_{\alpha,\gamma}f(x)=-(\Delta f)(x)\,,\quad x\in \R^3 \setminus \gamma \,.
$$
This construction yields a self-adjoint operator which gives meaning to the formal expression \eqref{formal}. It is useful to keep in mind that $\alpha$, in contrast to $\tilde\alpha$ in \eqref{formal}, is the `true' coupling constant. The perturbation is not additive, in particular, its absence corresponds to $\alpha=\infty$.

Our interest here concerns the discrete spectrum of $H_{\alpha,\gamma}$. As in the papers quoted above, an efficient way to study it is to employ Birman-Schwinger principle. For the starlike interaction supported on $\gamma$ we introduce the operator-valued matrix
 \begin{equation} \label{BSoperator}
Q_{\kappa,\gamma}:= [T_{\kappa,\gamma }^{ij}]_{i,j=1}^N
 \end{equation}
acting in $\bigoplus_{i=1}^N L^2 ([0,L])$, where $T_{\kappa,\gamma }^{ij}:\: L^2 ([0,L])\to L^2 ([0,L])$ are  integral operators with the kernels
 \begin{equation} \label{BSoperator2}
\left\{
  \begin{array}{ll}
    T_{\kappa;s,t} (|\bar\gamma_i - \bar\gamma_j |^2):= G_{\kappa} (|\gamma_i(s) - \gamma_j(t)|) \;\;& \hbox{if} \;\;i\neq j \\[1em]
    G^{\mathrm{reg}}_\kappa (\gamma_i(s)- \gamma_i(t)) \;\;& \hbox{if} \;\;i=j
  \end{array}
\right.
 \end{equation}
Here $G_{\kappa}$ is the integral kernel of $(-\Delta +\kappa^2)^{-1}$ in $L^2(\R^3)$, explicitly
\begin{equation}\label{eq-kernelG}
 G_{\kappa }(x,x') = \frac{1}{4\pi}\frac{\e^{-\kappa |x-x'|}}{|x-x'|}\,,
\end{equation}
and $G^{\mathrm{reg}}_\kappa$ is the regularized kernel with the logarithmic singularity removed as described in \eqref{T_ii_reg} below. For the sake of simplicity we will write $T_{\kappa }^{ij}=T_{\kappa, \gamma }^{ij}$ if there is no risk of confusion.

The Birman-Schwinger principle allowing us to rephrase the investigation of $\sigma_\mathrm{disc}(H_{\alpha,\gamma})$ as analysis of the operator $Q_{\kappa,\gamma}$ can be expressed concisely as
\begin{equation}\label{BSprinciple}
  f\in \ker(\alpha-Q_{\kappa,\gamma}) \, \Leftrightarrow \, H_{\alpha,\gamma }g_{\kappa }= -\kappa^2 g_{\kappa} \quad\mathrm{where }\;\; g_{\kappa }=G_{\kappa }\ast f \,.
\end{equation}
In particular, one can infer from here the positivity of the ground state eigenfunction using the following claim which is obtained by mimicking the argument of Lemma~4.1 in \cite{EK08}:

\begin{lemma} \label{l:positive}
Let $\epsilon_\gamma$ denote the principal eigenvalue of $H_{\alpha,\gamma}$, then any element $f\in \ker(Q_{\sqrt{-\epsilon_\gamma}, \gamma}-\alpha)$ is a multiple of a unique positive function.
\end{lemma}
\noindent Finally, let us mention the dependence of the spectrum on the arm length~$L$.
\begin{lemma} \label{l:monotone}
The eigenvalues of $H_{\alpha,\gamma}$ are monotonously decreasing functions of $L$.
\end{lemma}
\begin{proof}
Using \eqref{eq-bc3} it is easy to see that a scaling transformation, $x\mapsto x'=x\zeta$ with $\zeta\in\R_+$ leads to an operator which is unitarily equivalent to that corresponding to the original star with the scaled coupling constant,
\begin{equation}\label{scaling}
 \alpha' = \alpha - \frac{1}{2\pi} \ln\zeta\,.
\end{equation}
It is well known that the eigenvalues of $Q_{\kappa,\gamma}$ are continuously increasing functions of energy \cite{Kr53}, hence the claim follows from \eqref{BSprinciple}.

\end{proof}

\section{Existence of eigenvalues}\label{s:existence}
\setcounter{equation}{0}

Since our problem concerns the principal eigenvalue we have first ask about the conditions which ensure that the discrete spectrum of $H_{\alpha,\gamma}$ is nonvoid. We consider separately the finite and infinite star cases starting with  $L<\infty$.

\subsection{Finite stars}\label{ss:finite}

It is straightforward to check that $\sigma_\mathrm{ess}(H_{\alpha,\gamma})=\R_+$ holds for any $\alpha\in\R$ and $L<\infty$, hence we have to search for the negative spectrum.

\begin{theorem} \label{th-existenceev}
For a fixed $L>0$ we have $\sigma_{\mathrm{disc}}(H_{\alpha,\gamma}) \neq \emptyset$ provided $H_{\alpha, \gamma_i}$ corresponding to the `star' of a single segment $\gamma_i\subset\gamma$ has at least one negative eigenvalue.
\end{theorem}

\noindent Before coming to the proof we need a couple of auxiliary statements.

\begin{lemma} \label{le-supTii}
$\sup \sigma (T^{ii}_{\kappa}) \to -\infty$ holds as $\kappa \to \infty.$
\end{lemma}
\begin{proof}
As indicated above the action of $T^{ii}_{\kappa}$ is expressed by means of the regularized kernel,
\begin{align}
  (T^{ii}_\kappa f)(s) =& \int_0^L G^{\mathrm{reg}}_\kappa(\gamma_i(s)- \gamma_i (t))\,f(t)\,\mathrm{d}t \nonumber \\[.3em]
  =& \,\lim_{d\to 0} \Big( \frac{1}{4\pi} \int_0^L \frac{\e^{-\kappa ((s-t)^2+d^2 )^{1/2}}}{((s-t)^2+d^2 )^{1/2}}\, f(t)\,\mathrm{d}t +\frac{1}{2\pi} f(s) \ln d \Big)\,. \label{T_ii_reg}
\end{align}
The right-hand side can be rewritten by means of Fourier transformation \cite{BL77} as
$$
T^{ii}_\kappa f = \mathcal{F}^{-1} \big( -\ln (p^2+\kappa^2)^{1/2} +\psi (1)\big) \mathcal{F}f\,,
$$
where $\psi(1)<0$, and therefore there is a number $\kappa_0$ such that for any $\kappa >\kappa_0$ we have
$$
(T^{ii}_\kappa f,f) \le -\ln \kappa\, \|f\|^2\,,
$$
which completes the proof.
\end{proof}

Next we have to estimate the norm of the non-diagonal elements $T^{ij}_\kappa $.

\begin{lemma} \label{le-T12norm}
Let $\phi_{ij}$ be the angle between $\gamma_i$ and $\gamma_j$, $i\neq j$. Then
\begin{subequations}
\begin{equation}\label{eq-estT+1}
\|T^{ij}_{\kappa, \gamma }\| \leq \tau (\phi_{ij})\,,
\end{equation}
where $(0,\pi] \ni \phi_{ij}\mapsto \tau  ( \phi_{ij} )$ is a continuously decreasing function of $\phi_{ij}$ which satisfies
\begin{equation}\label{eq-estT+2}
\tau (\phi_{ij}) \le \frac{\sqrt{2}}{4\pi}\, \big|\ln (1-\cos\phi_{ij})\big|+ \Oo(1) \quad \text{as }\;\; \phi_{ij}\to 0+\,.
\end{equation}
\globallabel{eq-estT+}
\end{subequations}
\end{lemma}
\begin{proof} In the following we write the distance appearing at the right-hand side of \eqref{distance} as
$$
\rho (s,t):=|\gamma _i (s) - \gamma _j (t)|= (s^2+t^2 - 2st\cos \phi_{ij})^{1/2}
$$
without indicating the fixed indices $i,j$. We start with the estimate
\begin{align*} 
  |(T^{ij}_{\kappa, \gamma }f_i, f_j)| &= \frac{1}{4\pi } \Big|\int_{0}^{L} \int_{0}^{L}
\frac{e^{-\kappa \rho (s,t)}}{\rho (s,t)}\, \overline{f_i (s)}f_j (t)\, \mathrm{d}s\mathrm{d}t \Big| \\
  &\leq \frac{1}{4\pi}\int_{0}^{L} \int_{0}^{L}
\frac{1}{\rho (s,t)}\, |f_i (s)f_j (t)|\, \mathrm{d}s\mathrm{d}t\,.
\end{align*}
With later purpose in mind we extend the function $f_i$ as follows,
$$
f_i ^{\mathrm{ex}}(s)= \left\{
  \begin{array}{ll}
    f_i(s) & \hbox{for } s\in [0,L]\\[.3em]
    0 & \hbox{for } s\in (L,L']
  \end{array}
\right.
$$
where $L':=\sqrt{2}L$ and we use the radial system of coordinates $(r,\theta)$ in the plane determined by $\gamma_i$ and $\gamma_j$ to parametrize the quarter-disc $B_{L'}= \{(s=r\cos \theta, t=r\sin \theta  ):\: r\in [0,L'] \,,\:\theta \in [0, \pi/2]\}$. This allows us to rewrite the above estimate as
\begin{eqnarray}\label{eq-T12a}
 |(T^{ij}_{\kappa, \gamma }f_i , f_j )| \leq \frac{1}{4\pi }
\int_{B_{L'}}
\frac{|f_i ^{\mathrm{ex}}(r\cos\theta)f_j ^{\mathrm{ex}}(r\sin\theta)|}
{(1-\cos\phi_{ij} \sin 2\theta)^{1/2}}\, \mathrm{d}r\mathrm{d}\theta \,.
\end{eqnarray}
We assess the right-hand side of (\ref{eq-T12a}) using Schwarz inequality,
\begin{eqnarray} \nonumber
&& \hspace{-2em} \int_{B_{L'}}
\frac{|f_i^{\mathrm{ex}}(r\cos \theta )f_j ^{\mathrm{ex}}(r\sin \theta )| }{(1-\cos\phi_{ij} \sin 2 \theta )^{1/2}} \,\mathrm{d}r\mathrm{d}\theta \\ \nonumber && \hspace{-1em} \leq
\int_{0}^{\pi /2} \left( \frac{1}{1-\cos \phi_{ij} \sin 2 \theta }  \int_{0}^{L'} |f_i^{\mathrm{ex}} (r\cos \theta )|^2 \mathrm{d}r
\int_{0}^{L'} |f_j^{\mathrm{ex}}(r'\sin \theta )|^2 \mathrm{d}r' \right)^{1/2}  \mathrm{d}\theta  \\ \nonumber  && \hspace{-1em} = \int_{0}^{\pi /2} \left(
\frac{1}{\cos\theta \sin \theta (1-\cos \phi_{ij} \sin 2 \theta)}
\int_{0}^{L'\cos \theta } |f_i^{\mathrm{ex}}(t )|^2 \mathrm{d}t  \right.\\ \nonumber  &&\times \left.
 \int_{0}^{L' \sin \theta } |f_j ^{\mathrm{ex}} (t' )|^2 \mathrm{d}t' \right)^{1/2}  \mathrm{d}\theta  \\ \label{eq-estint1} && \hspace{-1em} \leq
\sqrt{2}\,\mathrm{I}_{\phi_{ij}}\, \Big(\int_{0}^{L' } |f_i^{\mathrm{ex}}(t )|^2\mathrm{d}t \int_{0}^{L'} |f_j^{\mathrm{ex}}(t )|^2 \mathrm{d}t\Big)^{1/2} = \sqrt{2}\, \mathrm{I}_{\phi_{ij}}  \|f_i\| \|f_j\| \,,
\end{eqnarray}
where
$$
\mathrm{I}_{\phi_{ij}} := \int_{0}^{\pi /2}
\frac{1}{ \sqrt{\sin 2\theta\, (1- \cos \phi_{ij}  \sin 2 \theta  )}}\, \mathrm{d}\theta\,.
$$
Note that $\mathrm{I}_{\phi_{ij}}$ is decreasing as a function of $\phi_{ij}$ and to show that we can identify it with $\frac{4\pi }{\sqrt{2}}\,\tau(\phi_{ij})$ we have to estimate it for small values of $\phi_{ij}$. For definiteness we suppose that $\phi_{ij}<\frac13\pi$ and  rewrite $\mathrm{I}_{\phi_{ij}}$ as
\begin{eqnarray} \label{eq-estint1a}
  \mathrm{I}_{\phi_{ij}} = \frac{1}{2} \int_{0}^\pi \frac{\sqrt {1-\cos \phi_{ij} \sin\theta'}}{\sqrt{\sin \theta'}}\, \mathrm{d}\theta' +  \frac{\cos\phi_{ij}}{2} \int_{0}^\pi \frac{\sqrt{\sin\theta'}}{\sqrt{1-\cos\phi_{ij} \sin\theta'}}\,\mathrm{d}\theta',
\end{eqnarray}
then for the first integral in the above expression we get
\begin{equation} \label{first_est}
\left|\frac{1}{2} \int_{0}^\pi \frac{\sqrt {1-\cos \phi_{ij} \sin\theta' }}{\sqrt{\sin \theta' }}\, \mathrm{d}\theta' \right| \leq \frac{1}{\sqrt{2}}\,
\int_{0}^\pi \frac{1}{\sqrt{\sin \theta' }}\, \mathrm{d}\theta' = \frac{\pi}{2}\,,
\end{equation}
while to the second component of (\ref{eq-estint1a}) we apply trigonometric identities,
\begin{align*}
\mathrm{J}_{\phi_{ij}} &:= \frac{\cos \phi_{ij} }{2} \int_{0}^\pi \frac{1}{\sqrt {1-\cos \phi_{ij} \sin \theta'}} \,\mathrm{d}\theta' = \cos \phi_{ij} \int_{-\pi /2 }^0
\frac{1}{\sqrt {1-\cos\phi_{ij} \cos \theta'}}\, \mathrm{d}\theta' \\
&=  \cos \phi_{ij}  \int_{-\pi /2 }^0
\frac{1}{\sqrt {1-\cos \phi_{ij} +2 \cos\phi_{ij} \sin ^2 \frac{\theta'}{2} }}\,\mathrm{d}\theta'\\
&=\frac{2 \cos \phi_{ij} }{\sqrt{2 \cos\phi_{ij} }}
\int_{-\pi /4 }^0
\frac{1}{\sqrt { \varsigma +  \sin ^2 t }}\,\mathrm{d}t\,,
\end{align*}
where
\begin{equation}\label{subst}
\varsigma \equiv \varsigma_{ij}:= \frac{1-\cos \phi_{ij} }{2\cos \phi_{ij} }\in \big(0,\textstyle{\frac12}\big)\,.
\end{equation}
The quantity $\mathrm{J}_{\phi_{ij}}$ is thus expressed through an elliptic integral and we have to find its behavior as $\phi_{ij}\to 0 $ which means $\varsigma \to 0$. To this aim we employ the substitution $\eta =\frac{\varsigma}{2} + \sin ^2 t$ which leads to
$$
\mathrm{J}_{\phi_{ij}} = \sqrt{\frac{ \cos \phi_{ij} }{2 }}
\int_{\frac{\varsigma +1}{2} }^{\frac{\varsigma}{2}}
\frac{1}{\sqrt{\eta^2 - \left(\frac{\varsigma}{2}\right)^2 } \sqrt { 1-\left(\eta -\frac{\varsigma}{2}\right)}}\,\mathrm{d}\eta \,;
$$
this expression can be estimated as
$$
\mathrm{J}_{\phi_{ij}} \le \sqrt{\cos \phi_{ij}}\,
\int_{\frac{\varsigma +1}{2} }^{\frac{\varsigma}{2}}
\frac{1}{\sqrt{\eta^2 - \left(\frac{\varsigma}{2}\right)^2 }}\,\mathrm{d}\eta =
\sqrt{\cos \phi_{ij}} \,\ln \frac{\varsigma}{\varsigma+1+\sqrt{2\varsigma+1}}\,.
$$
Returning to the original variable from \eqref{subst} and taking into account that the remaining part of the estimation expression is bounded by \eqref{first_est} we arrive at the desired conclusion.
\end{proof}

\noindent\emph{Proof of Theorem~\ref{th-existenceev}:} According to the assumption there is a $\kappa_0 >0$ and a corresponding (normalized) vector $f_i$ such that
$$
T^{ii}_{\kappa_0} f_i =\alpha f_i\,;
$$
without loss of generality we may suppose that it is largest eigenvalue of $T^{ii}_{\kappa_0}$ for which $f_i$ can be chosen positive by Lemma~\ref{l:positive}. Consider next a vector $f\in\bigoplus_{i=1}^N L^2 ([0,L])$ the $i$-th component is the said function $f_i$. If the other components are also positive, we have
$$
(Q_{\kappa_0,\gamma}f,f) > (T^{ii}_{\kappa_0}f_i, f_i) = \alpha
$$
due to the positivity of the kernel \eqref{eq-kernelG} which means that $\sup\sigma(Q_{\kappa_0,\gamma})>\alpha$. Furthermore, we note that
$$
[0,\infty) \ni \kappa \mapsto (Q_{\kappa,\gamma}f, f)
$$
is a continuous decreasing function as mentioned already in the proof of Lemma~\ref{l:monotone}. Now we used the above lemmata: we have $\sup\sigma(T^{ii}_\kappa )\to -\infty$ as $\kappa \to \infty$ by Lemma~\ref{l:monotone} while the non-diagonal operators $T^{ij }_\kappa$ remain bounded in this limit for fixed angles between the edges, and consequently, we have
$$
\sup \sigma(Q_{\kappa,\gamma}) \to -\infty \quad \mathrm{as}\quad \kappa \to \infty\,.
$$
In combination with $\sup\sigma(Q_{\kappa_0,\gamma})>\alpha$ this implies that there is a $\kappa_0'\geq \kappa_0$ and a vector $f$ such that
$$
Q_{\kappa_0',\gamma}f =\alpha f\,,
$$
which is what we have set out to prove. \hfill $\Box$

Combining this result with the claim about eigenvalues of $H_{\alpha,\gamma}$ describing the interaction supported by a segment obtained in \cite{EK08} by Dirichlet bracketing we arrive at the following conclusion:
\begin{corollary}
$\sigma_{\mathrm{disc}}(H_{\alpha,\gamma}) \neq \emptyset$ holds whenever $L> 2\pi\, \e^{2\pi \alpha-\psi(1)}$.
\end{corollary}

\subsection{Infinite stars}\label{ss:infinite}

The case $L=\infty$ has to be considered separately because the essential spectrum is then different. One cannot use directly the result from~\cite{EK02}, not even if $N=2$, because the interaction support there was supposed to be smooth, however, the argument can be easily modified.

\begin{theorem}
For any infinite star we have
\begin{equation}\label{eq-ess}
\inf\sigma_{\mathrm{ess}} (H_{\alpha,\gamma}) \ge \epsilon^\alpha\,,
\end{equation}
and moreover, with the exception of the situation when $N=2$ and $\gamma $ is a straight line,
\begin{equation}\label{eq-disc}
\sigma_{\mathrm{disc}} (H_{\alpha,\gamma}) \neq \emptyset \,.
\end{equation}
\end{theorem}
\begin{proof}
To check the first claim we consider semi-cylinders $\mathcal{C}_i$ of radius $d$ centered at $\gamma_i$ with a flat circular `bottom' $\tilde{\mathcal C}_i$ the boundary $\partial\tilde{\mathcal C}_i$ of which is a circle on the sphere being a boundary of a ball $\mathcal{B}\subset\R^3$ of radius $\varrho$ centered at the origin; it is clear that to a given $d$ one can choose $\varrho$ large enough to ensure that $\partial\tilde{\mathcal C}_i \cap \partial\tilde{\mathcal C}_j = \emptyset$ for $i\ne j$. We denote $\mathcal{D} := \mathcal{B} \setminus \big(\cup_{i=1}^N \mathcal{C}_i\big)$ and $\mathcal{J} := \R^3 \setminus \left( \mathcal{D} \cup \big( \cup_{i=1}^N \mathcal{C}_i\big) \right)  $. Then the entire space $\R^3$ is the union $\mathcal{J} \cup \mathcal{D}  \cup \big( \cup_{i=1}^N \mathcal{C}_i\big)$. The corresponding Neumann bracketing then yields a lower bound to $\sigma_{\mathrm{ess}} (H_{\alpha,\gamma})$. The parts of the spectrum referring to $\mathcal{D}$ and $\mathcal{J}$ are discrete and positive, respectively, and it remains to analyze the spectrum of $H_{\alpha,\gamma}\upharpoonright \mathcal{C}_i$ which define embedding of $H_{\alpha,\gamma}$ to $\mathcal{C}_i$ with Neumann boundary conditions. According to \cite[Lemma~3.6]{EK04} there is a $c>0$ such that
$$
\inf\sigma_{\mathrm{ess}} (H_{\alpha,\gamma}\upharpoonright \mathcal{C}_i) \ge \epsilon^\alpha - \mathrm{e}^{-cd}
$$
holds as $d\to\infty$, and since $d$ can be chosen arbitrarily large \eqref{eq-ess} follows.

To establish (\ref{eq-disc}) we denote by $\breve{\gamma}$ the excluded case, a straight line, and use a comparison with the operator $H_{\alpha,\breve{\gamma}}$ the spectrum of which is obviously $[\epsilon^\alpha,\infty)$ corresponding to $\sigma(Q_{\kappa, \breve\gamma}) = (-\infty,s_\kappa]$ where $s_\kappa:= \frac{1}{2\pi} \big(\psi(1) - \ln\frac{\kappa}{2}\big)$. By assumption one can always choose a pair of non-parallel arms of $\gamma$, without loss of generality we may suppose that they are $\gamma_1$ and $\gamma_2$. Choosing a trial function $\phi$ sufficiently `spread' along the broken line $\gamma_1 \cup \gamma_2$ in analogy with \cite[Lemma~5.2]{EK02} one can achieve that
$$ 
(Q_{\kappa,\gamma_1 \cup \gamma_2}\, \phi,\phi ) > s_\kappa\,.
$$ 
The natural decomposition  $\phi = \phi_1 \oplus \phi_2$ with $\phi_i \in L^2 ([0,\infty))$ allows us then to construct the trial function $\phi^{\mathrm{ext}}=(\phi_1,\phi_2, 0,\dots,0)$ which gives
$$
(Q_{\kappa,\gamma}\,\phi^{\mathrm{ext}},\phi ^{\mathrm{ext}}) =
(Q_{\kappa,\gamma_1 \cup \gamma_2}\, \phi,\phi) > s_\kappa\,.
$$
The latter means in view of \eqref{BSprinciple} that $\inf\sigma(H_{\alpha,\gamma}) < \inf\sigma(H_{\alpha,\breve \gamma}) = \epsilon^\alpha$, and combining this result with (\ref{eq-ess}) we arrive at~(\ref{eq-disc}).
\end{proof}

\begin{remark}
{\rm It is also easy to construct a suitable Weyl sequence showing that $\sigma_{\mathrm{ess}} (H_{\alpha,\gamma}) = [\epsilon^\alpha,\infty)$ but we will not need this result in the following.
}
\end{remark}

\section{Non-existence of the discrete spectrum} \label{s:nonexistence}
\setcounter{equation}{0}

Despite the interaction we consider is strongly singular, $H_{\alpha,\gamma}$ shares the property of three-dimensional Schr\"odinger operators concerning the absence of weakly bound states for regular potentials. For a fixed finite star we expect this to happen if the $\delta$-interaction is sufficiently weak, i.e. $\alpha$ large enough; by the unitary equivalence mentioned in the proof of Lemma~\ref{l:monotone} the same happens for a fixed $\alpha$ and $L$ small enough. In \cite{EK08} we proved that for a segment $\gamma=\gamma_i$ of length $L$ one has $\sup T^{ii}_{\kappa} <\frac{1}{2\pi}\ln \frac{L}{4}$ which in view of \eqref{BSprinciple} means that the discrete spectrum is void provided
 $$
 \frac{1}{2\pi }\ln \frac{L}{4}<\alpha\,.
 $$
For a star-shaped support this result generalizes in the following way:

\begin{theorem} \label{th-noev}
There is a $\,C>0$ such that $\sigma_{\mathrm{disc}} (H_{\alpha,\gamma})= \emptyset$ holds if
\begin{equation}\label{eq-noev}
  \frac{N}{2\pi } \ln \frac{L}{4}+ \sum_{i\neq j} \Big( \frac{\sqrt{2}}{4\pi}|\ln(1-\cos\phi_{ij})|+C\Big)< \alpha\,.
\end{equation}
\end{theorem}
\begin{proof}
Consider any $f=(f_1,...,f_N)\in\bigoplus_{i=1}^N L^2 ([0,L])$. In view of the mentioned result from \cite{EK08} and Lemma~\ref{le-T12norm} we can estimate the upper threshold of the operator $Q_{\kappa,\gamma}$ as
$$
\sup_f\, (Q_{\kappa,\gamma }f, f) = \sum_{ij} \sup_{f_i,f_j} (T_\kappa ^{ij}f_i,f_j)
   \leq \frac{N}{2\pi} \ln \frac{L}{4} + \sum_{i\neq j}\Big( \frac{\sqrt{2}}{4\pi}\,|\ln(1-\cos\phi_{ij})|+C\Big)
$$
for some $C>0$; the suprema in the above formula are taken over all functions belonging to the domains of corresponding operators. This, in view of \eqref{BSprinciple}, yields the condition \eqref{eq-noev}.
\end{proof}

\section{Small-angle asymptotics} \label{s:small}
\setcounter{equation}{0}

Our stated goal is the optimization of the principal eigenvalue of $H_{\alpha,\gamma}$. Before coming to it we want to show that such a stationary point cannot be a minimum. Let us look in detail at the case of a two-arm star, $\gamma = (\gamma_1,\gamma_2)$, with the angle $\phi_{12}=\phi$ between the edges. We are going to show that for $\phi$ small enough the operator has any prescribed finite number of eigenvalues and the $k$-th one escapes to $-\infty$ as $\phi\to 0$. Moreover, we present also a lower bound to such eigenvalues:

\begin{theorem} \label{th-smallangle2}
For a finite $L>0$ there is a family of eigenvalues $E_k$ of $H_{\alpha,\gamma}$, $k=1,2,\dots$, that satisfy the inequalities
\begin{equation}\label{eq-estimEk}
  E_k^- \leq E_k \leq E_k^+\,,
\end{equation}
as $\phi \to 0$, where
$$
E_k^+ := -\frac{2\sqrt{2}\,\e ^{-2\pi\alpha +2\psi(1)}}{L} \frac{1}{(1-\cos\phi )^{1/2}} + \Big(\frac{\pi k}{L}\Big)^2 + o(\phi)
$$
and
$$
E_k^- := -4\e^ {2(-2\pi C -2\pi\alpha +\psi(1))}\frac{1}{1-\cos \phi } + \Big(\frac{\pi k}{L}\Big)^2 + o(\phi)
$$
with the constant $C$ of Theorem~\ref{th-noev}.
\end{theorem}
\begin{proof}
According to \eqref{BSoperator} and \eqref{BSprinciple} the spectral condition for $H_{\alpha,\gamma}$ reads
\begin{equation}\label{eq-searc}
 \sum_{j=1}^2 T^{ij}_\kappa f_j = \alpha f_i\,,\quad i=1,\,2\,.
\end{equation}
The symmetry of the system implies that the eigenvectors of
 \begin{equation}\label{eq-BSop2}
Q_{\kappa , \gamma } = [T^{ij }_\kappa ]_{i,j =1}^2
\end{equation}
are symmetric or antisymmetric with respect to the permutation of the edges, $\tilde{f}=(f,\pm f)$. We note first that the antisymmetric case is irrelevant for our present purpose. Indeed, the restriction of an antisymmetric eigenfunction of $H_{\alpha,\gamma}$ refers to the halfspace problem with the segment emanating from the Dirichlet boundary. By the bracketing argument \cite[Sec.~XIII.15]{RS} the respective eigenvalue is not smaller than the one referring to the same segment in the full space. However, the latter is independent of $\phi$, and moreover, by Lemma~\ref{l:monotone} it is not smaller than $\epsilon^\alpha$ given by \eqref{2D ev}.

Hence we may consider $\tilde{f}=(f,f)$ for which (\ref{eq-searc}) reduces to the form
$$
T^{11}_{\kappa  }f+ T^{12}_{\kappa  }f = \alpha f\,.
$$
Applying the result of \cite[Lemma~3.1]{EK08} we find
\begin{align}\label{eq-T11}
 T^{11}_{\kappa }f (s)&= \int_{0}^L G^{\mathrm{reg}}_\kappa (\gamma_1(t )-\gamma_1(s))\,f(t)\, \mathrm{d}t \\ \nonumber
 &= \frac{1}{4\pi}\Big(\int_{0}^L \frac{f(t)-f(s)}{|t-s|}\,\mathrm{d}t + f(s)\ln 4s(L-s) \Big)  \\  \nonumber
& \quad + \frac{1}{4\pi} \int_{0}^L
\Big(\frac{\e^{-\kappa |s-t|}}{|s-t|}  -\frac{1 }{|s-t|} \Big)f(t)\,\mathrm{d}t \,.
\end{align}
On the other hand, the action of $T^{12}_{\kappa}$ can be expressed using \eqref{BSoperator2}; denoting as before $\rho (s,t)= \sqrt{s^2+t^2 -2st\cos \phi}$ we can rewrite in a form similar to \eqref{eq-T11}, namely
\begin{align}\label{eq-T12b}
 T^{12}_{\kappa  }f (s)&= \int_{0}^L G_\kappa (\gamma_1(t)-\gamma_2(s))\,f(t)\,\mathrm{d}t \\ \nonumber
 &= \frac{1}{4\pi}\left(\int_{0}^L \frac{f(t)-f(s)}{\rho(t,s)}\,\mathrm{d}t +f(s)\int_{0}^L\frac{1}{\rho(t,s)}\, \mathrm{d}t \right. \\ \nonumber  & \quad +\left.\int_{0}^L \Big(\frac{\e^{-\kappa\rho(t,s)}}{\rho(t,s)} -\frac{1}{\rho(t,s)}\Big)\,f(t)\,\mathrm{d}t\right) \,.
\end{align}
We need some estimates of the quantities appearing in these expressions:
\begin{align}
& \int_{0}^L\frac{1}{\rho(t,s)} \,\mathrm{d}t = \ln \big(\sqrt{(L-s)^2 +2Ls (1-\cos\phi)}+(L-s)
+(1-\cos\phi)s \big) \nonumber \\[.5em] \label{eq-estim1} & \qquad -\ln s -\ln (1-\cos\phi)
\;\geq\; \ln 2(L-s) -\ln s(1-\cos\phi)\,.
\end{align}
Furthermore, using the fact that $\rho (s,t)$ is symmetric we have
\begin{equation}\label{eq-est2}
  \int_{(0,L)^2} \frac{(f(t)-f(s))f(s)}{\rho(t,s)}\,\mathrm{d}t \mathrm{d}s = -\frac{1}{2}
\int_{(0,L)^2} \frac{(f(t)-f(s))^2}{\rho(t,s)}\,\mathrm{d}t \mathrm{d}s \leq 0\,.
\end{equation}
In combination with
\begin{equation}\label{eq-estq}
\rho  (s,t) \geq |s-t|
\end{equation}
this gives
\begin{equation}\label{eq-est3}
  0 \ge \int_{(0, L)^2} \frac{(f(t)-f(s))f(s)}{\rho (t,s)}\,\mathrm{d}t \mathrm{d}s \geq \int_{(0, L)^2} \frac{(f(t)-f(s))f(s)}{|t-s|}\, \mathrm{d}t \mathrm{d}s\,.
\end{equation}
Next we consider the function $\xi:\,\R_+ \to \R$ defined by $\xi (x)=\frac{\e^{-\kappa x}}{x} -\frac{1}{x}$  which is easily seen to be increasing for any positive $\kappa$. This monotonicity together with (\ref{eq-estq}) gives
\begin{equation}\label{eq-est4}
  \frac{\e^{-\kappa\rho(t,s)}}{\rho(t,s)} -\frac{1}{\rho(t,s)} \geq
\frac{\e^{-\kappa|s-t|}}{|s-t|} -\frac{1}{|s-t|}\,.
\end{equation}
The  estimates (\ref{eq-estim1}), (\ref{eq-est3}), (\ref{eq-est4}) in combination with (\ref{eq-T11}) yield
\begin{align} \nonumber
(f, (T^{11}_{\kappa } &+ T^{12}_{\kappa })f) \\ \nonumber
&\geq 2 (f, T^{11}_{\kappa } f )
+\frac{1}{4\pi } \left( -\int_{0}^L \ln 2s\, |f(s)|^2 \,\mathrm{d}s -\ln(1-\cos\phi) \|f\|^2\right)\\ \label{eq-T12-} &\geq
2 (f, T^{11}_{\kappa }f )
+\frac{1}{4\pi } \big( -\ln 2L^2  -\ln(1-\cos\phi) \big) \|f\|^2
\end{align}
Next we introduce the operators
\begin{equation}\label{eq-exT+}
T^- := 2  T^{11}_{\kappa }
- \frac{1}{4\pi } \ln (1-\cos \phi ) - \frac{1}{4\pi } \ln 2L
\end{equation}
and
\begin{equation}\label{eq-exT-}
T^+ := T^{11}_{\kappa} - \frac{1}{4\pi} \ln (1-\cos \phi ) +C\,, \quad T:= T^{11}_{\kappa } + T^{12}_{\kappa}\,,
\end{equation}
where $C$ is the constant analogous to that in Theorem~\ref{th-noev} (and implicitly in Lemma~\ref{le-T12norm}). Putting together the estimates (\ref{eq-T12-}) and (\ref{eq-estT+}) we arrive at the inequalities
\begin{equation}\label{eq-sandwich}
 T^- \leq T \leq T^+
\end{equation}
that hold for $\phi$ small enough in the form sense. Let next $\tau_k^\pm(\kappa)$ stand for the discrete eigenvalues of the operators $T^\pm$ and $\tau_k (\kappa )$ for the discrete eigenvalues of $T$, all ordered in the same way. As a consequence of (\ref{eq-sandwich}) and the min-max principle we have
$$ 
\tau_k^- (\kappa )\leq  \tau_k (\kappa ) \leq \tau_k^+(\kappa )\,.
$$ 
Furthermore, let $\kappa_k^\pm $ stand for the solutions of $\tau_k (\kappa)^\pm =\alpha$ and let $\kappa_k$ refer similarly to the solution of $\tau_k (\kappa)=\alpha$ which determines by \eqref{BSprinciple} the eigenvalues of $H_{\alpha,\gamma}$. Using the inequality (\ref{eq-sandwich}) together with the fact that $\kappa \mapsto \tau_k$ and $\kappa \mapsto  \tau_k^\pm$ are continuous decreasing functions we conclude that
$$
\kappa_k^- \leq \kappa_k \leq \kappa_k^+\,.
$$
In view of (\ref{eq-exT+}) and (\ref{eq-exT-}) the estimating numbers $\kappa _k^\pm$ correspond to the eigenvalues
$E_k^\pm := - (\kappa ^\mp _k)^2 $ of the operator with $\delta$ interaction supported by a straight segment of the length $L$ and coupling constants
\begin{equation}\label{eq-alpha+}
\alpha_+ = \frac{1}{4\pi}\ln (1-\cos\phi) +C +\alpha
\end{equation}
and
\begin{equation}\label{eq-alpha-}
\alpha_- =  \frac{1}{2}\left(\frac{1}{4\pi}\ln (1-\cos\phi) +\frac{1}{4\pi} \ln 2L^2+\alpha \right)\,.
\end{equation}
On the other hand, arguing as in \cite{EK16}, we find that the above mentioned eigenvalues referring to a coupling constant $\beta\in\R$ behave for large $\beta$ as
$$
\kappa _k =  \Big( 4 e^{2(-2\pi \beta +\psi (1) )}+\frac{(\pi k)^2}{L^2} + o(\beta^{-1})\Big)^{1/2}\,;
$$
inserting (\ref{eq-alpha+}) and (\ref{eq-alpha-}) for $\beta$ in the above expression we get the claim.
\end{proof}

In a similar way one could treat a star $\gamma$ with $N>2$ arms, $\gamma = \gamma_1 \cup...\cup \gamma_N$. If $N-1$ arms are fixed, without loss of generality being supposed to be $\gamma_2,..., \gamma_N$, and the remaining one moves in such a way that the angle $\phi_{12}$ between it and $\gamma_2$ tends to zero, one can again conclude that the spectral threshold escapes to $-\infty$ noting that all the contributions to $Q_{\kappa,\gamma}$ remain bounded except the one coming from the closing angle which explodes in the same way as in the previous proof. Finally, by Lemma~\ref{l:monotone} the conclusion extends to infinite stars, $L=\infty$.

\section{Energy optimization} \label{s:optim}
\setcounter{equation}{0}

Now we can finally pass to our main topic, the question about the star configuration for which the principal eigenvalue of $H_{\alpha,\gamma}$ is maximal. To begin with, we have to recall several notions from algebraic combinatorics \cite{BB09, CK07} inspired by the old and difficult Thomson's problem \cite{Th1904}.

Consider $N$ points $\{x_i\}_{i=1}^N$ placed on a unit sphere $S^2$. They are said to form an $M-$ spherical design if for any polynomial function $S^2\ni x\mapsto p(x)$ of total degree at most $M$ its mean over $\{x_i\}$ coincides with the mean over the sphere,
$$
\int_{S^2} p(x)\,\mathrm{d}x = \frac{1}{N} \sum_{i=1}^N p(x_i )\,.
$$
Suppose further that $m$ denotes the number of the different inner product between the points, then $\{x_i\}_{i=1}^N$ is called a sharp configuration if it is $2m-1$ spherical design. A deep result proved in \cite{CK07}, see also \cite{BB09}, says that any sharp configuration is universally optimal, in other words, it minimizes \emph{any} potential energy described by a strictly completely monotonous function, $(-1)^k f^{(k)} \geq 0$ for all $k\in\N$. This result is valid for sphere configurations in any dimension. In the three-dimensional situation we are interested in here there are five sharp configurations as listed in Table~1 of~\cite{CK07}:
\begin{itemize}
\setlength{\itemsep}{-2pt}
\item $N=2$: antipodal points
\item $N=3$ simplex with the inner product $-\frac12$
\item $N=4$: tetrahedron, i.e. simplex with the inner product $-\frac13$
\item $N=6$: octahedron, i.e. cross polytope with the inner products $-1,\,0$
\item $N=12$: icosahedron with the inner products $-1,\, \pm 5^{-1/2}$.
\end{itemize}
We denote by $\{\bar\sigma_j\}_{j=1}^N$ the sharp configuration of $N$ points, and furthermore, $\sigma$ will be an  $N$-arms star with the arms $\sigma_i$ of the length $L$, emanating from the origin and such that they, or their halfline extensions, contain $ \bar\sigma_j$. The key element of our discussion is the following lemma:
\begin{lemma}
Suppose that $\gamma$ is $N-$ arm star with the corresponding points distribution $\{\bar\gamma_i\}_{i=1}^N$ on the unit sphere. The inequality
\begin{equation}\label{eq-monotT}
  \sum_{1\le i<j\le N} T_{\kappa;s,t}(|\bar\gamma_i - \bar\gamma_j |^2) \geq
  \sum_{1\le i<j\le N} T_{\kappa;s,t}(|\bar\sigma_i - \bar\sigma_j |^2 )
\end{equation}
holds for any $s,t \in [0,L]$, and moreover, (\ref{eq-monotT}) becomes equality if and only if $\{\gamma_j\}_{j=1}^N$ is congruent with a sharp configuration.
\end{lemma}
\begin{proof}
One has to establish that the function $(0,4]\ni x \mapsto T_{\kappa;s,t} (x) \in \R_+$ is strictly completely monotonous. The function in question equals
$$
T_{\kappa;s,t} = \frac{\e^{-\kappa \sqrt{a+bx}}} {4\pi \sqrt{a+bx}}
$$
with $a=(s-t)^2$ and $b=-2st$, and a straightforward computation gives
$$
 T_{\kappa;s,t}^{(k)}=\frac{(-1)^k }{4\pi} P_{k+1}\left(\frac{ 1}{\sqrt{a+bx}} \right) \frac{x\e^{-\kappa \sqrt{a+bx} }}{2\sqrt{a+bx}}\,,
$$
where $x\mapsto P_n(x)$ is a positive polynomial of $n$-th degree. This establishes the strictly complete monotonicity and the claim is then a direct consequence of Theorem~1.2 in \cite{CK07}.
\end{proof}

Recall that $\epsilon_\gamma$ is the principal eigenvalue of $H_{\alpha,\gamma}$. We know from Lemma~\ref{l:positive} that the corresponding eigenfunction is a multiple of a positive function. If the star refers to a sharp configuration we can say more:
\begin{lemma} \label{le-symetric}
Let $N\in \{2,3,4,6,12\}$, then any $\tilde f_{\sigma} \in \  \ker (Q_{\sqrt{-\epsilon_\sigma},\gamma}-\alpha) $ is symmetric function in the sense that $\tilde f_\sigma =(f_\sigma,...,f_\sigma)$ with a $f_\sigma \in L^2 ([0,L])$.
\end{lemma}
\begin{proof}
The argument is similar to the one used in \cite{EL17}: using the fact that the distances between the points of $\bar\sigma$ are fixed, one concludes that the subspace of symmetric functions in $\bigoplus_{i=1}^N L^2 ([0,L])$ is invariant under $Q_{\kappa,\gamma}$ and its orthogonal complement consisting of function with zero mean. Since $\tilde f_{\sigma}$ is positive by Lemma~\ref{l:positive} it cannot belong to the latter.
\end{proof}

Now we are in position to state our main result:
\begin{theorem} Let $N\in \{2,3,4,6,12\}$, then the
energy $\epsilon_\gamma $ assumes the unique maximum for $\gamma$ congruent with $\sigma$.
\end{theorem}
\begin{proof}
Using (\ref{eq-monotT}) together with the fact that the diagonal elements of $Q_{\kappa,\gamma}$ do not depend on the angles,
$$
(f , T_{\kappa, \gamma   }^{ii} f)_{L^2 ([0,L])} =
(f, T_{\kappa, \sigma  }^{ii} f)_{L^2 ([0,L])}\,,
$$
and Lemma~\ref{le-symetric} we get
\begin{align*}
   & \sup Q_{\kappa,\gamma} \geq (Q_{\kappa,\gamma} \tilde f_\sigma, \tilde f_\sigma) \\
   & \hspace{1em} \geq \sum_{1\le i<j\le N} \int_{L\times L}
   T_{\kappa;s,t}((|\bar\gamma_i -\bar\gamma_j |^2))f_\sigma(s) f_\sigma(t)\,\mathrm{d}s
   \mathrm{d}t+\sum_{i=1}^N
   (f_\sigma, T_{\kappa,\gamma}^{ii} f_\sigma )_{L^2 ([0,L])}
  \\
   & \hspace{1em} \geq \sum_{1\le i<j\le N} \int_{L\times L}
   T_{\kappa ;s,t }((|\bar\sigma_i -\bar\sigma_j |^2)) f_\sigma(s) f_\sigma(t)\,\mathrm{d}s
   \mathrm{d}t +\sum_{i=1}^N (f_\sigma, T_{\kappa,\sigma}^{ii} f_\sigma)_{L^2 ([0,L])} \\
   &= \hspace{1em} \sup Q_{\kappa,\sigma}\,,
\end{align*}
and the inequality is sharp unless $\bar\gamma$ is congruent with $\bar\sigma$.
\end{proof}

\begin{remarks}
{\rm (a) Note that the argument works both for any edge lengths giving rise to a discrete spectrum including infinite stars, $L=\infty$.
\\
(b) Finding optimal configurations for other values of $N$ is no doubt a difficult problem. We note that while for an infinite star the answer is independent of $\alpha$ due to the self-similar character of $\gamma$, cf. Lemma~\ref{l:monotone}, this may not be the case if $L<\infty$.
}
\end{remarks}

\subsection*{Acknowledgements}

The research was supported by the Czech Science Foundation (GA\v{C}R) under Grant No. 17-01706S and by the European Union within the project CZ.02.1.01/0.0/0.0/16 019/0000778.


\end{document}